\documentclass[journal,twocolumn,letterpaper]{IEEEtran}

\usepackage{graphicx,cite,epsfig,amssymb,amsmath}
\usepackage{pifont}
\usepackage{stmaryrd}
\usepackage{bbding}
\usepackage{subfigure}
\usepackage{algorithm}
\usepackage{algorithmic}
\usepackage{slashbox,multirow}
\usepackage{mathrsfs}
\usepackage{latexsym}
\usepackage{indentfirst}
\usepackage{fmtcount}
\usepackage{cite}
\usepackage{float}
\usepackage{morefloats}

\usepackage{color}

\newcommand{\tabincell}[2]{\begin{tabular}{@{}#1@{}}#2\end{tabular}}

\allowdisplaybreaks[4]

\newtheorem{theorem}{Theorem}

\newtheorem{lemma}{Lemma}

\begin{document}
\title{Gossip-based Information Spreading in Mobile Networks}
\author{Huazi~Zhang, \textit{Student Member, IEEE}, Zhaoyang~Zhang, \textit{Member, IEEE}, and Huaiyu~Dai, \textit{Senior Member, IEEE}
\thanks{
Huazi~Zhang (e-mail: {\tt hzhang17@ncsu.edu}) is with the Department of Information Science and Electronic Engineering, Zhejiang University, and is currently a visiting Ph.D. student in North Carolina State University, under the joint supervision of Zhaoyang Zhang and Huaiyu Dai.
Zhaoyang~Zhang (e-mail: {\tt ning\_ming@zju.edu.cn})
is with the Department of Information Science
and Electronic Engineering, Zhejiang University, China.
Huaiyu~Dai
(e-mail: {\tt huaiyu\_dai@ncsu.edu}) is with the Department of
Electrical and Computer Engineering, North Carolina State
University, USA.}

\thanks{This work was supported in part by the National Key Basic Research Program of China (No. 2012CB316104), the National Hi-Tech R\&D Programm of China (No. 2012AA121605), the National Science Foundation of China under grant No. 61371094, the Supporting Program for New Century Excellent Talents in University (NCET-09-0701), and the US National Science Foundation under Grants ECCS-1002258 and CNS-1016260.}}
\maketitle

\begin{abstract}
In this paper, we analyze the effect of mobility on information spreading in geometric networks through natural random walks. Specifically, our focus is on epidemic propagation via \emph{mobile gossip}, a variation from its static counterpart. Our contributions are twofold. Firstly, we propose a new performance metric, mobile conductance, which allows us to separate the details of mobility models from the study of mobile spreading time. Secondly, we utilize geometrical properties to explore this metric for several popular mobility models, and offer insights on the corresponding results. Large scale network simulation is conducted to verify our analysis.
\end{abstract}


\setlength\arraycolsep{1pt}

\IEEEpeerreviewmaketitle

\section{Introduction}
\subsection{Motivation}
Mobile networks receive increasing research interest recently; mobile ad hoc networks (MANET) and vehicular ad hoc networks (VANET) are two prominent examples. In many real world networks, an interesting application is to broadcast the information from some source node to the whole network. For wireless ad hoc and sensor networks, a node triggered by the event of interest may want to inform the whole network about the situation as quickly as possible. For social networks, rumors and stories are forwarded by people via different communication media. In these and many other applications, how fast a message can be spread to the whole network is of particular interest as opposed to the general network throughput.

Mobility introduces challenges as well as opportunities. It is known to improve the network throughput as shown in \cite{Mobility-throughput}. However, its effect on information spreading is still not very well understood. In mobile networks, will the information spreading speed up or slow down? How can we quantify the potential improvement or degradation due to mobility? How may the different mobility patterns affect the information spreading? These problems are of major importance and deserve further study.

\subsection{Related Works}
Information spreading in static networks has been well
studied in literature \cite{Gossip}, \cite{Gossip2}. Gossip
algorithm, dated back to \cite{Early-gossip}, is a simple but
effective fully distributed information spreading strategy, in
which every node randomly selects \emph{only one} of its neighbors for message
exchange during the information spreading process.
Therefore, gossip algorithms assume certain advantages over other widely adopted information spreading approaches such as flooding. Furthermore, gossip algorithms can achieve near-optimal performance for a class of static network graphs including random geometric graphs \cite{Gossip}. It is also found that the spreading time in static networks is closely related to the geometry of a network, named ``conductance" \cite{Conduct}, \cite{Conduct2}, which essentially represents the bottleneck for information exchange within a network.

Mobile networks have drawn significant research interest in
recent years. Traditionally, mobility is viewed as a negative
feature as it adds additional uncertainty to wireless networks, and
incurs more challenges in channel estimation. Recently, mobility has been revisited for its potential to improve network performance. In
the seminal works \cite{Mobility-throughput}, mobility is shown to
significantly increase the sum-throughput of the network under the fully random mobility model; later the study is extended to the one-dimensional mobility model\cite{One-dim-mobility-1}.
Subsequently, the throughput-delay tradeoff is further investigated
in the context of mobile ad-hoc networks \cite{Throu-delay-Gamal-IT2006,Throu-delay,Throu-delay-Sharma-Infocom2006,Throu-delay2,Throu-delay3,Throu-delay4}.

There has been extensive study on both information
spreading and mobility of networks, separately. Recently, some interesting analytical results for information spreading in dynamic wireless networks have emerged. In particular, the delay of epidemic-style routing is studied in \cite{XiaolanZhang},
\cite{PoissonMeetingP2PDelay} assuming exponential distributions for inter-meeting times. The scaling properties of information propagation between a pair of nodes in large mobile wireless networks are explored in \cite{ZKongMobiHoc2008}, for the constrained i.i.d. mobility and discrete-time Brownian motion models. Subsequently, an upper bound of the information propagation speed for the flooding mechanism is derived in \cite{Jacquet2010} for the random walk mobility model, with the emphasis on the sparse networks (in particular when the node density tends to 0). Some other recent progress includes \cite{PSSS11}, \cite{PPPU11}, where again only \emph{random-walk like mobility models} \cite{TCampWCMC2002} are considered. An exception is the work \cite{Spreading-MEG}, in which an upper bound of the flooding time is derived in terms of \emph{node-expansion properties} of a general stationary Markovian evolving graph. However, this approach requires a node transmission range \emph{above the connectivity threshold}. When extended to the sparse scenario \cite{AClementiICALP2009}, the previous expansion technique fails to work, and a set of probabilistic results are developed for a special case of the random walk mobility model.
In \cite{Gossip-mobility}, the impact of mobility on the average consensus problem is investigated, where again the transmission range is required above the connectivity threshold, and only memoryless (time-independent) mobility models are considered.

\subsection{Summary of Contributions}
Motivated by the existing study, we intend to develop a more general analytical framework for gossip-based information spreading in mobile networks which can address various types of mobility patterns, and admits wider applicability concerning transmission radius and network connectivity. The main contributions of this paper are summarized below.

\begin{enumerate}
\item
Based on a ``move-and-gossip" information spreading model, we propose a new metric, mobile conductance, which represents the capability of a mobile network to conduct information flows. Mobile conductance is dependent not only on the network structure, but also on the mobility patterns. Facilitated by the definition of mobile conductance, a general result on the mobile spreading time is derived for a class of mobile networks modeled by stationary Markovian evolving graphs, with a less stringent requirement on node transmission range and network connectivity.

\item
We evaluate the mobile conductances for various mobility models, including fully random mobility, partially random mobility, velocity constrained mobility, one-dimensional and two-dimensional area constrained mobility, and offer insights on the results. The results are summarized in Table~\ref{result-table}\footnote{We follow the standard notations. Given non-negative functions $f(n)$ and $g(n)$: $f(n) = O(g(n))$ if there exists a positive constant $c_1$ and an integer $k_1$ such that $f(n) \leq c_1 g(n)$ for all $n \geq k_1$;
$f(n) = \Omega(g(n))$ if there exists a positive constant $c_2$ and an integer $k_2$ such that $f(n) \geq c_2 g(n)$ for all $n \geq k_2$;
$f(n) = \Theta(g(n))$ if both $f(n) = O(g(n))$ and $f(n) = \Omega(g(n))$ hold;
$f(n) = o(g(n))$ if there exists a positive constant $c_3$ such that $f(n) \leq c_3 g(n)$ for all $n \geq c_3$;
$f(n) = \omega(g(n))$ if there exists a positive constant $c_4$ such that $f(n) \geq c_4 g(n)$ for all $n \geq c_4$.}. In particular, the study on the fully random mobility model reveals that the potential improvement in information spreading time due to mobility is dramatic: from $\Theta \left( {\sqrt n } \right)$ to $\Theta \left( \log n \right)$. We have also carried out large scale simulations to verify our analysis.

\end{enumerate}

The rest of this paper is organized as follows. System, mobility, and information spreading models are presented in Section \ref{problem-formulation}. Mobile conductance is defined in Section \ref{main}, and a general result on the mobile spreading time is derived. In Section \ref{app}, mobile conductances of several popular mobility models (see \cite{Velocity-mobility,One-dim-mobility-2,Two-dim-mobility,Sunlei2013CRBlackhole,Sunlei2013MobileCRN} and the references therein) are evaluated and corroborated by simulation results, leading to some interesting insights. Finally Section \ref{conclusion} concludes the work.

\begin{table}[!t]
\renewcommand{\arraystretch}{1.1}
\caption{Conductances of Different Mobility Models} \label{result-table} \centering
\begin{tabular}{|c|c|}
\hline
Static Conductance & ${\Phi _s} = \Theta \left( {\sqrt {\frac{{\log n}}{n}} } \right).$ \\
\hline\hline
Mobility Model & Mobile Conductance ${\Phi_m}$ \\
\hline
Fully Random & $\Theta \left( 1 \right).$ \\
\hline
Partially Random & ${\left( {\frac{{n - k}}{n}} \right)^2}{\Phi _s} + \frac{{k\left( {2n - k} \right)}}{{2{n^2}}}.$ \\
\hline
\multirow{3}{*}{\tabincell{c}{Velocity\\ Constrained}} & \multirow{3}{*}{$\Theta \left( {\max \left( {{v_{\max }},r} \right)} \right)$} \\
&\\
&\\
\hline
\multirow{3}{*}{\tabincell{c}{Area Constrained:\\ One-Dim}} & \multirow{3}{*}{$\frac{n_V^2+n_H^2}{n^2}{\Phi _s} + {\frac{{{n_V}{n_H}}}{{n^2}}} .$} \\
&\\
&\\
\hline
\multirow{3}{*}{\tabincell{c}{Area Constrained:\\ Two-Dim}}  & \multirow{3}{*}{$\Theta \left( {\max \left( {{r_c},r} \right)} \right)$} \\
&\\
&\\
\hline
\end{tabular}
\end{table}

\section{Problem Formulation}\label{problem-formulation}

\subsection{System Model}
We consider an $n$-node mobile network on a unit square $\tilde \Psi$, modeled as a time-varying graph $G_{t}\triangleq(V,E_t)$ evolving over discrete time steps. The set of nodes $V$ are identified by the first $n$ positive integers $[n]$. One key difference between a mobile network and its static counterpart is that, the locations of nodes change over time according to certain mobility models, and so do the connections between the nodes represented by the edge set $E_t$.

It is assumed that the moving processes of all nodes $\{X_i(t), t \in \mathrm{\mathbf{N}}\}$, $i \in [n]$, are independent stationary Markov chains, each starting from its stationary distribution with the transition distribution $q_i$, and collectively denoted by $\{\mathbf{X}(t), t \in \mathrm{\mathbf{N}}\}$ with the joint transition distribution $Q=\prod\limits_{i = 1}^n {{q_i}}$. While not necessary, we assume the celebrated random geometric graph (RGG) model \cite{RGG-book} for the initial node distributions for concreteness (particularly in Section \ref{app}), i.e., $G_0 = G(n,r)$, where $r$ is the common transmission range assumed for all nodes. Under many existing random mobility models such as those considered in \cite{Mobility-throughput,ZKongMobiHoc2008,Spreading-MEG,Velocity-mobility,Gossip-mobility} and in this work, nodes will maintain the uniform distribution on the state space $\tilde \Psi$ over the time. Two nodes are neighbors if they are within distance $r$ at some time instant. The speed of node $i$ at time $t$ is defined by $v_i(t)=|X_i(t+1)-X_i(t)|$, assumed upper bounded by $v_{max}$ for all $i$ and $t$. We also assume a less stringent requirement for network connectivity, as described below: for an arbitrary node subset $S' \subset V$, it is not totally isolated from its complement $\bar{S'}$ after the move in the expectation sense\footnote{More specifically, our result requires $\mathbb{E}_Q[N_{S'}(t+1)]>0$; see \eqref{mobile-conductance-approx}.}; for RGG this implies $v_{\max}+r = \Omega\left({\sqrt {\frac{{\log n}}{n}} }\right)$\footnote{This requirement is already a relaxation as compared to $r = \Omega\left({\sqrt {\frac{{\log n}}{n}} }\right)$ demanded for static networks.}.

\subsection{Mobility Model}\label{mobility-model}
For notation convenience, the unit square is discretized into a grid with a sufficiently high resolution $\delta$: $\Psi=\{ (i\delta,j\delta)|0\leq i,j \leq \lfloor 1/\delta \rfloor\}$. Denote $q^{i}_{xy}=q_i(X_i(t+1)=y|X_i(t)=x)$, $x,y \in \Psi$, $\forall i$, as a generic element of the transition matrix $Q$. The following mobility models are considered in this study:

Fully Random Mobility \cite{Mobility-throughput}:  $X_i(t)$ is uniformly distributed on $\Psi$ and i.i.d. over time. In this case $v_{\max}=\Theta(1)$, $q^{i}_{xy}=1/|\Psi|, \ \forall i, \ \forall x,y \in \Psi$. This idealistic model is often adopted to explore the largest possible improvement brought about by mobility.

Partially Random Mobility: $k$ randomly pre-selected nodes are mobile, following the fully random mobility
model, while the rest $n-k$ nodes stay static. This is one generalization of the fully random mobility model.

Velocity Constrained Mobility \cite{Velocity-mobility, Spreading-MEG}: This is another generalization of the fully random mobility model, with the node speed bounded by an arbitrary $v_{\max}=O(1)$. In this case, $q^{i}_{xy}=1/|\mathcal{C}(x)|, \ \forall i$ and $\forall y \in \mathcal{C}(x)$, where $\mathcal{C}(x)=\{y\in \Psi||y-x|\leq v_{\max }\}$; and $q^{i}_{xy}=0$, otherwise.

One-dimensional Area Constrained Mobility \cite{One-dim-mobility-1,One-dim-mobility-2}: In this model, the mobile nodes move either vertically (named V-nodes) or horizontally (named H-nodes), reminiscent of trains or automobiles moving on the railways or city streets.
It is assumed that both V-nodes and H-nodes are uniformly and randomly distributed on $\Psi$, and the the mobility pattern of each
node is ``fully random" on the corresponding one-dimensional path. Let $x\triangleq(x_\alpha,x_\beta)\in \Psi$ and $y\triangleq(y_\alpha,y_\beta)\in \Psi$. For a V-node, $q^{i}_{xy}=1/(\lfloor 1/\delta \rfloor+1), \ \forall i$ and $\forall y \in \mathcal{V}(x)$, where $\mathcal{V}(x)=\{y\in \Psi|y_\alpha=x_\alpha\}$; and $q^{i}_{xy}=0$, otherwise. The transition probability for an H-node is similarly defined.

Two-dimensional Area Constrained Mobility\cite{ZKongMobiHoc2008,Two-dim-mobility}: In this model, each node $i$ has a \emph{unique} home point $i_h$, and moves around the home point within a disk of radius $r_c$ uniformly and randomly.
The home points are fixed, independently and uniformly distributed on $\Psi$. Here $q^{i}_{xy}=1/K, \ \forall i$ and $\forall x,y \in \Psi_i$, where $\Psi_i=\{y\in \Psi||y-i_h|\leq r_c\}$ while $K$ is the number of grid points in a circle of radius $r_c$; and $q^{i}_{xy}=0$, otherwise. $r_c$ is also called \emph{mobility capacity}. This model may simulate the patrol scenarios by police or automatic mobile agents.

\subsection{Spreading Model: Move-and-Gossip}
We consider the problem of single-piece information dissemination through a natural randomized gossip algorithm in \cite{Gossip}. The extension to the multi-piece dissemination problem readily follows and will be addressed in our future work. In contrast to the static case, there is an additional moving process mixed with the gossip process. In this study, we adopt a move-and-gossip model as shown in Fig.~\ref{Move-and-Gossip} to describe information spreading in a mobile network and facilitate our analysis. Specifically, each time slot is decomposed into two phases: each node first \emph{moves} independently according to some mobility model as discussed above, and then \emph{gossips} with \emph{one} of its \emph{new} neighbors. During the gossip step, it is assumed that each node \emph{independently} contacts one of its neighbors uniformly at random, and during each meaningful contact (where at least one node has the piece of information), the message is successfully delivered in either direction (through the ``push" or ``pull" operation).
\begin{figure}[h] \centering
\includegraphics[width=0.45\textwidth]{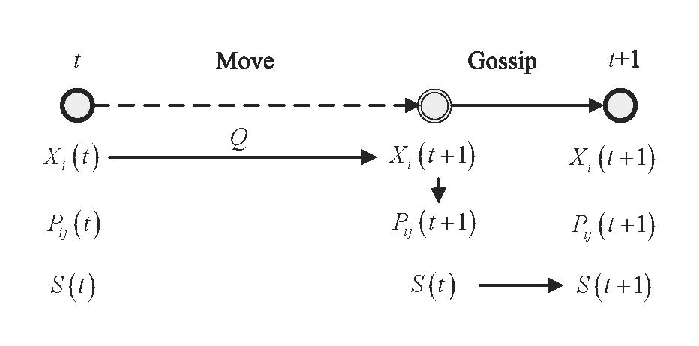}
\caption {Move-and-Gossip Spreading Strategy}
\label{Move-and-Gossip}
\end{figure}

Denote $S\left( t \right) \subset V$ as the set of nodes that have the message, at the \emph{beginning} of time slot $t$. Initially only the source node $s$ has the message, i.e. $S\left( 0 \right) = \{s\}$. A careful check of the move-and-gossip paradigm reveals the following unique features: the node position $X_i(t)$ changes in the middle of each time slot (after the move step), while $S(t)$ is updated at the end (after the gossip step). $P_{ij}(t+1)$ is used to denote the the probability that node $i$ contacts one of its \emph{new} neighbors $j \in {\cal N}_i(t+1)$ in the gossip step of slot $t$; for a natural randomized gossip, it is set as $1/|{\cal N}_i(t+1)|$ for $j \in {\cal N}_i(t+1)$, and $0$ otherwise. For a RGG $G(n,r)$, in the static case, $P_{ij}$ is on the order of $P(n,r)=\Theta\left(\frac{1}{{n\pi {r^2}}}\right)$ when $j \in {\cal N}_i$\cite{RGG-book,Gossip}. In the mobile case, the stochastic matrix $P(t)=\left[P_{ij}(t)\right]_{i,j=1}^{n}$, which collects the contact probabilities over all node pairs, changes over time (in terms of connections) governed by the transition matrix $Q$ of the homogeneous Markov chain $\{\textbf{X}(t)\}$, but the values of non-zero $P_{ij}(t)$'s remain on the order of $P(n,r)$.


Our performance metric is the $\varepsilon$-dissemination time, defined as (where * stands for static or mobile):
\begin{align}\label{spreading-time}
 T_{*} \left( \varepsilon  \right)
 \triangleq \mathop {\sup }\limits_{s \in
V} \inf \left\{ {t:\Pr \left( {\left|S\left( t \right)\right| \ne n\left|
{S\left( 0 \right) = \left\{ s \right\}} \right.} \right) \le
\varepsilon } \right\}.
\end{align}

\section{Mobile Conductance}\label{main}
\subsection{Preliminaries on Static Networks}
We first recall some relevant results in static networks. According to \cite{Gossip},
the static spreading time scales as
\begin{align}\label{static-spreading-time}
{T_{static}}\left( \varepsilon  \right) = O\left( {\frac{{\log n + \log {\varepsilon ^{ - 1}}}}{{{\Phi _s}}}} \right),
\end{align}
where $\Phi _s$ is the static conductance defined as
\begin{align}\label{static-conduct}
\Phi _s = & \mathop {\min
}\limits_{\scriptstyle S \subset V , \left| S \right| \le n/2 \hfill} \left( {\frac{{\sum\limits_{i \in S,j \in \overline S } {P_{ij} } }}{{\left| S \right|}}} \right) \nonumber \\ \overset{\text{(uniform)}}{\doteq} & \mathop {\min
}\limits_{\scriptstyle S \subset V , \left| S \right| \le n/2 \hfill} \left( {\frac{{P\left( n,r \right){N_S}}}{{\left| S \right|}}} \right),
\end{align}
where the second expression holds for the RGG in the order sense, and ${N_S}$ is the number of connecting edges between set $S$ and its complement $\bar S$. Note that ${N_S}$ is a constant for a given set $S$ in the static case, but becomes a random variable in the mobile case when the nodes in $S(t)$ and $\overline{{S}(t)}$ move at each time step.

It has been shown that the conductance for a static random geometric graph scales as $\Theta \left( r \right)$ \cite{Conduct} for $r = \Theta \left( \sqrt {\frac{\log n}{n}}\right)$, and the static spreading time scales as
\begin{align*}
{T_{static}} = O\left( {\frac{{\log n}}{{\sqrt {{{\log n} \mathord{\left/
 {\vphantom {{\log n} n}} \right.
 \kern-\nulldelimiterspace} n}} }}} \right) \approx O\left( {\sqrt n } \right).
\end{align*}

It is worth mentioning that the above result is actually tight in the order sense. The network radius is on the order of $\Theta \left( 1 \right)$, and the distance of one-hop transmission is $\Theta \left( {\sqrt {\frac{{\log n}}{n}} } \right)$. Thus, the minimal number of hops is on the order of $\Theta \left( \sqrt {\frac{n}{{\log n}}}\right) \approx  \Theta \left( {\sqrt n } \right)$. This indicates that the spreading time in the static network scales as $\Theta \left( {\sqrt n } \right)$.

\subsection{Mobile Conductance and Mobile Spreading Time}
Conductance essentially determines the static network bottleneck in information spreading. Intuitively, node movement introduces dynamics into the network structure, thus can facilitate the information flows. In this work we define a new metric, \emph{mobile conductance}, to measure and quantify such improvement.


\emph{Definition:}
The mobile conductance of a stationary Markovian evolving graph with transition distribution $Q$ is defined as:

\begin{align}
 {\Phi _m}\left( Q \right)
&\triangleq \mathop {\min }\limits_{\scriptstyle {S'(t) \subset V}\atop
\scriptstyle {\left| {S'\left( t \right)} \right| \le n/2} }   \left\{  \mathbb{E}_Q \left( {\frac{{\sum\limits_{ i \in S'\left( t \right), j \in \overline {S'\left( t \right)} } {{P_{ij}}\left( {t + 1}\right)} }}{{\left| {S'\left( t \right)} \right|}}} \right) \right\} \label{mobile-conductance}\\ \label{mobile-conductance-approx}
&\overset{\text{(uniform)}}{\doteq} \mathop {\min }\limits_{\scriptstyle {S'(t) \subset V}\atop
\scriptstyle {\left| {S'\left( t \right)} \right| \le n/2} }   \left\{ {\frac{{P\left(n, r \right)}}{{\left| {S'\left( t \right)} \right|}}\mathbb{E}_Q \left[ {{N_{S'}}\left( {t + 1} \right)} \right]} \right\},
\end{align}
where ${S'\left( t \right)}$ is an arbitrary node set with size no larger than $n/2$, and $N_{S'}\left( t+1 \right)$ is the number of connecting edges between ${S'\left( t \right)}$ and $\overline {S'\left( t \right)}$ after the move.


\emph{Remarks:} Some explanations for this concept are in order.

\emph{1)} Similar to its static counterpart, we examine the cut-volume ratio for an arbitrary node set ${S'\left( t \right)}$ (not the message set) at the beginning of time slot $t$. Different from the static case, due to the node motion ($X_i(t)\rightarrow X_i(t+1)$ in Fig. \ref{Move-and-Gossip}), the cut structure (and the corresponding contact probabilities $\{P_{ij}(t)\}$) changes. Thanks to the stationary Markovian assumption, its expected value (conditioned on ${S'\left( t \right)}$) is well defined with respect to the transition distribution $Q$. Minimization over the choice of ${S'\left( t \right)}$ essentially determines the bottleneck of information flow in the mobile setting.

\emph{2)}
As mentioned above, for a RGG $G(n,r)$, the stochastic matrix $P(t)=\left[P_{ij}(t)\right]_{i,j=1}^{n}$ changes over time, but the values of non-zero $P_{ij}(t+1)$'s remain on the same order of $P(n,r)$ given that nodes are uniformly distributed. This allows us to focus on evaluating the number of connecting edges between ${S'\left( t \right)}$ and $\overline {S'\left( t \right)}$ \emph{after} the move: ${N_{S'}}\left( t+1 \right)= \sum_{i \in S'\left( t \right),j \in \overline {S'\left( t \right)} } {{I_{ij}}\left( t+1 \right)}$.\footnote{${I_{ij}}\left( {t + 1} \right) \triangleq \left\{ {\begin{array}{*{20}{c}}
{1,j \in {\cal N}_i\left( {t + 1} \right)}\\
{0,j \notin {\cal N}_i\left( {t + 1} \right)}
\end{array}} \right.$ is the indicator function for the event that node $i$ and $j$ become neighbors after the move and before the gossip step in slot $t$.} Therefore for network graphs where nodes keep uniform  distribution over the time, mobile conductance admits a simpler expression (\ref{mobile-conductance-approx}).

\emph{3)} This definition may naturally be extended to the counterpart of $k$-conductance in \cite{Gossip}, with the set size constraint of $n/2$ in (\ref{mobile-conductance}) replaced by $k$, to facilitate the study of multi-piece information spreading in mobile networks.

Based on the above definition, we can obtain a general result for information spreading in mobile networks as shown below.
\begin{theorem}\label{theorem-mobile-spreading-time}
For a mobile network with mobile conductance $\Phi_m(Q)$,
the mobile spreading time scales as
\begin{equation}\label{mobile-spreading-time}
T_{mobile} \left( \varepsilon, Q  \right) = O\left( {\frac{{\log n + \log
\varepsilon ^{ - 1} }}{{\Phi _m(Q) }}} \right).
\end{equation}
\end{theorem}
\begin{proof}
We follow the standard procedure of the static counterpart (e.g., in \cite{Gossip}), with suitable modifications to account for the difference between static and mobile networks. Starting with $|S(0)|=1$, the message set $S(t)$ monotonically grows through the information spreading process, till the time $|S(t)|=n$ which we want to determine. The main idea is to find a good lower bound on the expected increment $|S(t+1)|-|S(t)|$ at each slot. It turns out that such a lower bound is well determined by the conductance of the network. Since the conductance is defined with respect to sets of size no larger than $n/2$, a two-phase strategy is adopted, where the first phase stops at $|S(t)|\leq n/2$. In the first phase, only the ``push" operation is considered (thus the upper bound on the spreading time is safe); while in the second phase, the emphasis is on the ``pull" aspect of the nodes in $\overline {S\left( t \right)}$ (whose size is no larger than $n/2$). Since the two phases are symmetric, we will only focus on the first one.

In the first phase, for each node $j \in \overline {S\left(t\right)}$, define a random variable $\Delta_j\left(t\right)$. If at least one node with the message moves to the $j$'s neighboring area in slot $t$ and ``pushes" the message to $j$ in the gossip step, one new member is added to the message set. We let $\Delta_j{\left( {t + 1} \right)}=1$ in this case, and $0$ otherwise. In the following, we will evaluate the expected increase $|S(t+1)|-|S(t)|$ conditioned on $S(t)$. The key difference between the static and mobile case is that, there is an additional move step in each slot; therefore, the expectation is evaluated with respect to both the moving and gossiping process. This is where our newly defined metric, mobile conductance, enters the scene and takes place of the static conductance. Specifically, due to the independent actions of nodes in $S(t)$ after the move, we have
%
\begin{align*}
\mathbb{E}\left[ {{\Delta_j}\left( {t + 1} \right)\left| {S\left( t \right)} \right.} \right]
=& \mathbb{E}_Q \left[ {1 - \prod\limits_{i \in S\left( t \right)} {\left( {1 - {P_{ij}}\left( {t + 1} \right)} \right)} } \right] \\
\ge &\mathbb{E}_Q \left[ {1 - \prod\limits_{i \in S\left( t \right)} {\exp \left( { - {P_{ij}}\left( {t + 1} \right)} \right)} } \right] \\
\ge& \frac{1}{2}\mathbb{E}_Q \left[ {\sum\limits_{i \in S\left( t \right)} {{P_{ij}}\left( {t + 1} \right)} } \right],
\end{align*}
where the first and the second inequalities are due to the facts of $1 - x < \exp \left( { - x} \right)$ for $x \ge 0$ and $1 - \exp \left( { - x} \right) \ge \frac{x}{2}$ for $0 \le x \le 1$, respectively. Then
%
\begin{align}\label{set-increase}
&\mathbb{E}\left[ {\left| {S\left( {t + 1} \right)} \right| - \left| {S\left( t \right)} \right|\left| {S\left( t \right)} \right.} \right]
=\sum\limits_{j \in \overline {S\left( t \right)} } {\mathbb{E}\left[ {{\Delta_j}\left( {t + 1} \right)\left| {S\left( t \right)} \right.} \right]}\nonumber \\
\ge &\frac{1}{2}\mathbb{E}_Q \left[ {\sum\limits_{i \in  {S\left( t \right)} ,j \in \overline {S\left( t \right)} } {{P_{ij}}\left( {t + 1} \right)} } \right]\nonumber \\
= &\frac{{\left| {S\left( t \right)} \right|}}{2}  \mathbb{E}_Q \left[ {\frac{{\sum\limits_{ i \in S\left( t \right), j \in \overline {S\left( t \right)} } {{P_{ij}}\left( {t + 1}\right)} }}{{\left| {S\left( t \right)} \right|}}} \right]\nonumber \\
\ge &\frac{{\left| {S\left( t \right)} \right|}}{2}\mathop {\min }\limits_{\scriptstyle S'\left( t \right) \subset V\atop
\scriptstyle \left| {S'\left( t \right)} \right| \le n/2} \left\{ \mathbb{E}_Q \left[ {\frac{{\sum\limits_{ i \in S\left( t \right), j \in \overline {S\left( t \right)} } {{P_{ij}}\left( {t + 1}\right)} }}{{\left| {S\left( t \right)} \right|}}} \right] \right\}\nonumber \\
= &\frac{{\left| {S\left( t \right)} \right|}}{2}{\Phi _m}\left( Q \right).
\end{align}
The form of \eqref{set-increase} is consistent with the counterpart in static networks \cite{Gossip}. Therefore, we can follow the same lines in the rest part of the proof.

%
\end{proof}

\section{Application}\label{app}
The general definition of mobile conductance allows us to separate the details of mobility models from the study of mobile spreading time. In this section, we will evaluate the mobile conductances of several popular mobility models.\footnote{In the following calculation, the resolution parameter $\delta$ in Section \ref{mobility-model} goes to 0.}

We will assume that the network instances follow the RGG model for concreteness, and evaluate (\ref{mobile-conductance-approx}). The main efforts in evaluation lie in finding the bottleneck segmentation (i.e., one that achieves the minimum in (\ref{mobile-conductance-approx})), and determining the expected number of connecting edges between the two resulting sets.
It is known \cite{Conduct} that for a static RGG $G(n,r)$, the bottleneck segmentation
is a bisection of the unit square, when $n$ is sufficiently large.
Intuitively, mobility offers the opportunity to escape from any bottleneck structure of the static network, and hence facilitates the spreading of the information. As will be shown below, fully random mobility destroys such a bottleneck structure, in that ${S'\left( t \right)}$ and $\overline {S'\left( t \right)}$ are fully mixed after the move; this move yields mobile conductance of $\Theta(1)$, a dramatic increase from static conductance  $\Theta \left( r \right)=\Theta \left( {\sqrt {\frac{{\log n}}{n}}} \right)$ \cite{Conduct}. Even for the more realistic velocity constrained model, part of the nodes from ${S'\left( t \right)}$ and $\overline {S'\left( t \right)}$ cross the boundary after the move and the connecting edges between the two sets are increased. The width of this contact region is proportional to $v_{\max}+r$.

\subsection{Evaluation of Several Mobility Models}

\subsubsection{Fully Random Mobility}
\begin{theorem}\label{fully-random}
In fully random mobile networks, the mobile conductance scales as $\Theta \left( 1 \right)$, and the corresponding mobile spreading time scales as $O\left( {\log n} \right)$.
\end{theorem}
\begin{proof}
 Since this mobility model is memoryless, for an arbitrary $S'\left( t
\right)$, the nodes in both $S'\left( t \right)$ and $\overline
{S'\left( t \right)}$ are uniformly distributed after the move, with density $|S'\left( t \right)|$ and $|\overline
{S'\left( t \right)}|$ respectively.
For each node in $S'\left( t \right)$, the size of its neighborhood
area is $\pi {r^2}$,
therefore, the expected number of contact pairs
\begin{align}
\mathbb{E}_Q \left[ {{N_{S'}}\left( {t + 1} \right)} \right]
 = \left| {S'\left( t \right)} \right|\left| {\overline {S'\left( t \right)} } \right|\pi {r^2}.
\end{align}

Noting that
\begin{equation*}\label{full-mobility-conductance}
{\frac{{P\left( n, r \right)}}{{\left| {S'\left( t \right)}
\right|}}\mathbb{E}_Q \left[
{{N_{S'}}\left( {t + 1} \right)} \right]}
 = \Theta \left( 1 \right),
\end{equation*}
regardless the choice of $S'(t)$ (with size no larger than $n/2$), we have $\Phi _m=\Theta(1)$. There is no bottleneck segmentation in this mobility model.
\end{proof}

\emph{Remark 1: }In the gossip algorithms, only the nodes with the message can contribute to the increment of $\left| {S\left( t \right)} \right|$. Consider the ideal case that each node with the message contacts a node without message in each step, which represents the fastest possible information spreading. We have the following straightforward arguments:
\begin{align*}
&\left| {S\left( {t + 1} \right)} \right| - \left| {S\left( t \right)} \right| \le \left| {S\left( t \right)} \right|\\
&\Rightarrow \left| {S\left( {t + 1} \right)} \right| \le 2\left| {S\left( t \right)} \right|\\
&\Rightarrow \left| {S\left( t \right)} \right| \le {2^t} = O\left( {{e^t}} \right).
\end{align*}

When $\left| {S\left( T \right)} \right|$ reaches $(1-\epsilon)n$,
the message has largely been spread to the whole network. Therefore,
$T_{mobile}(\epsilon) = \Omega \left( {\log n} \right)$ for
arbitrary constant $\epsilon$, and the optimal performance in
information spreading is achieved in the fully random model.

\emph{Remark 2: }While this model may not be practical, it reveals that the potential improvement on information spreading time due to mobility
is dramatic: from $\Theta \left( {\sqrt n } \right)$ to $\Theta
\left( \log n \right)$.

\subsubsection{Partially Random Mobility}


\begin{theorem}\label{theorem-partial} For the partially random mobility model, where
$k$ out of $n$ nodes are fully mobile, and the rest $n-k$
nodes stay static, the mobile conductance $\Phi _m={\left( {\frac{{n - k}}{n}} \right)^2}{\Phi _s} + \frac{{k\left( {2n - k} \right)}}{{2{n^2}}}$.
\end{theorem}

\begin{proof}
For each node that already has the message, say $i$, among all its
neighbors, there are on average $ \left( {n - k} \right)\pi r^2$
static nodes and $k\pi r^2$ mobile nodes. We denote the set of $k$
mobile (dynamic) nodes at time $t$ as $D\left( t \right)$, the set
of $n-k$ static nodes at time $t$ as $\overline {D\left( t \right)}$
and calculate the number of contacted pairs separately as follows.
\begin{align}\label{partial-mobility-conductance}
\mathbb{E}_Q \left[ {{N_{S'}}\left( {t + 1} \right)} \right]
= \mathbb{E}_Q \left[ {\begin{array}{*{20}{l}}
{\sum\limits_{i \in S'\left( t \right) \cap \overline {D\left( t \right)} ,j \in \overline {S'\left( t \right)}  \cap \overline {D\left( t \right)} } {{I_{ij}}\left( {t + 1} \right)} }\\
 + \sum\limits_{i \in S'\left( t \right) \cap D\left( t \right),j \in  \overline {S'\left( t \right)}  \cap D\left( t \right)} {{I_{ij}}\left( {t + 1} \right)}\\
{ + \sum\limits_{i \in S'\left( t \right) \cap D\left( t \right),j \in \cap \overline {S'\left( t \right)}  \cap \overline {D\left( t \right)} } {{I_{ij}}\left( {t + 1} \right)} }\\
 + \sum\limits_{i \in S'\left( t \right) \cap \overline {D\left( t \right)} ,j \in \overline {S'\left( t \right)}  \cap D\left( t \right)} {{I_{ij}}\left( {t + 1} \right)}
\end{array}} \right],
\end{align}
where the former two terms are the number of contact pairs \emph{within}
static nodes and mobile nodes, respectively, while the latter two terms are those \emph{between} static and mobile nodes.

The links within the static nodes remain unchanged after the move,
therefore $I_{ij} \left( {t + 1} \right) = I_{ij}$
for ${i \in S'\left( t \right) \cap \overline {D\left( t \right)},j \in
\overline {S'\left( t \right)}  \cap \overline {D\left( t \right)}}$. Since the
$k$ mobile nodes are fully random, the links involving the mobile
nodes (the last three terms) can be estimated similarly as in the fully random model. Putting together (with some reorganization), we have
\begin{align}
\mathbb{E}_Q \left[ {{N_{S'}}\left( {t + 1} \right)} \right]
&=
\mathbb{E}_Q \left[{\sum\limits_{\scriptstyle i \in S'\left( t \right) \cap \overline
{D\left( t \right)} ,\atop \scriptstyle j \in \overline {S'\left( t \right)}  \cap \overline {D\left(
t \right)}} {{I_{ij}}} }\right] \nonumber \\
& + \left(
\begin{array}{l}
\frac{{n - k}}{n}\left| {\overline {S'\left( t \right)} } \right|n\pi {r^2}\frac{{\left| {S'\left( t \right)} \right|}}{n}\frac{k}{n}\\
 + \frac{{n - k}}{n}\left| {S'\left( t \right)} \right|n\pi {r^2}\frac{{\left| {\overline {S'\left( t \right)} } \right|}}{n}\frac{k}{n}\\
 + \frac{k}{n}\left| {S'\left( t \right)} \right|n\pi {r^2}\frac{{\left| {\overline {S'\left( t \right)} } \right|}}{n}\frac{k}{n}
\end{array} \right).
\end{align}

According to the definition of mobile conductance,
\begin{align}\label{partial-mobility-conductance2}
{\Phi _m}
= & \mathop {\min }\limits_{\scriptstyle S'\left( t \right) \subset V ,\atop \scriptstyle \left| {S'\left( t \right)} \right| \le n/2}\left\{ {\frac{{P\left( n, r \right)}}{{\left| {S'\left( t \right)} \right|}}\sum\limits_{\scriptstyle i \in S'\left( t \right) \cap \overline {D\left( t \right)} ,\atop
\scriptstyle j \in \overline {S'\left( t \right)}  \cap \overline {D\left( t \right)}} {{I_{ij}}} } \right\} \nonumber \\
 +& \mathop {\min }\limits_{\scriptstyle S'\left( t \right) \subset V ,\atop \scriptstyle \left| {S'\left( t \right)} \right| \le n/2}\left\{ {2\frac{{n - k}}{n}\frac{{\left| {\overline {S'\left( t \right)} } \right|}}{n}\frac{k}{n} + \frac{k}{n}\frac{{\left| {\overline {S'\left( t \right)} } \right|}}{n}\frac{k}{n}} \right\}\nonumber \\
= & {\left( {\frac{{n - k}}{n}} \right)^2}{\Phi _s} + \frac{{k\left( {2n - k} \right)}}{{2{n^2}}}.
\end{align}

Note that the two minima are achieved simultaneously when $\left| {S'\left( t \right)} \right| = \left| {\overline {S'\left( t \right)} } \right| = \frac{n}{2}$.
\end{proof}
\emph{Remarks:} Since ${\Phi _s} = \Theta \left( {\sqrt {\log n / n}
} \right)$ and $\frac{1}{2}\frac{k}{n} < \frac{{k\left( {2n - k}
\right)}}{{2{n^2}}} < \frac{k}{n}$, the number of mobile nodes needs
to achieve $\omega \left( {\sqrt {n\log n} } \right)$ in order to
bring significant benefit over the static one. Partially random mobility
model is a mixture of the static network and fully random mobile network.
It can be seen that as $k$ grows, the mobile conductance increases:
as $k \to \Theta \left( n \right)$, $\Phi _m  \to \Theta \left( 1
\right)$.

\subsubsection{Velocity Constrained Mobility}
\begin{theorem}\label{theorem-velocity}
For the mobility model with velocity constraint $v_{max}$, the mobile conductance scales as $\Theta \left( \max \left( v_{\max},r \right) \right) $.
\end{theorem}



\begin{proof}
According to \emph{Lemma \ref{velocity-segmentation}} in Appendix
\ref{Lemma1}, the bottleneck segmentation between
${S'\left(t\right)}$ and $\overline {S'\left( t \right)}$ is the
straight line bisection of the unit square, and the density of nodes before and after
move is illustrated in Fig.~\ref{velocity}. For better illustration, darkness of the regions represents the density of nodes that belong to $S'\left( t \right)$. Before the move, the nodes in
$S'\left( t \right)$ and $\overline {S'\left( t \right)}$ are
strictly separated by a straight line border\footnote{Note that we have the flexibility to choose a network cut according to the definition of mobile conductance. }. After the move, with
some nodes in both $S'\left( t \right)$ and $\overline {S'\left( t
\right)}$ crossing the border to enter the other half, a mixture
strip of width $2 \times v_{\max }$ emerges in the middle of the
graph.

We take the center of the graph as the origin. Denote $\rho _{S'\left( t \right)} \left( l \right)$ and $\rho
_{\overline {S'\left( t \right)} } \left( l \right)$ as the density
of nodes before moving, and $\rho
'_{S'\left( t \right)} \left( l \right)$ and $\rho '_{\overline
{S'\left( t \right)} } \left( l \right)$ as the density of nodes
after moving, with $l$ the horizontal coordinate. As shown in
the upper subfigure of Fig.~\ref{velocity}, at time $t$, the nodes
in the circle of radius $v_{\max }$ have equal probabilities to move
to the center point at time slot $t+1$. Therefore, $\rho '_{S'\left(
t \right)} \left( l \right)$ is given by the proportion of
the dark area in the circle (thus is
uniform over the vertical line $x=l$).
$\rho '_{\overline {S'\left( t \right)} } \left( l
\right)$ can be obtained similarly.


After some derivation, we have
{\footnotesize{\begin{equation*}
{\frac{{{{\rho '}_{S'}}\left( l \right)}}{n} = \left\{ {\begin{array}{{ll}}
{1,} & {l  \leq   - {v_{\max }}},\\
\begin{array}{l}
\frac 1 \pi \arccos \left( {\frac{l}{{{v_{\max }}}}} \right)\\
 - \frac{l}{{\pi {v_{\max }}}}\sin \left( {\arccos \frac{l}{{{v_{\max }}}}} \right),
\end{array} & {- {v_{\max }} < l < {v_{\max }}},\\
{0,} & {l  \geq  {v_{\max }}},
\end{array}} \right.}\\
\end{equation*}}}

and
\begin{equation*}
{\frac{{{{\rho '}_{\overline {S'}}}\left( l \right)}}{n} = 1 - \frac{{{\rho '}_{S'}}\left( l \right)}{n}}.
\end{equation*}

\begin{figure}[h] \centering
\includegraphics[width=0.35\textwidth]{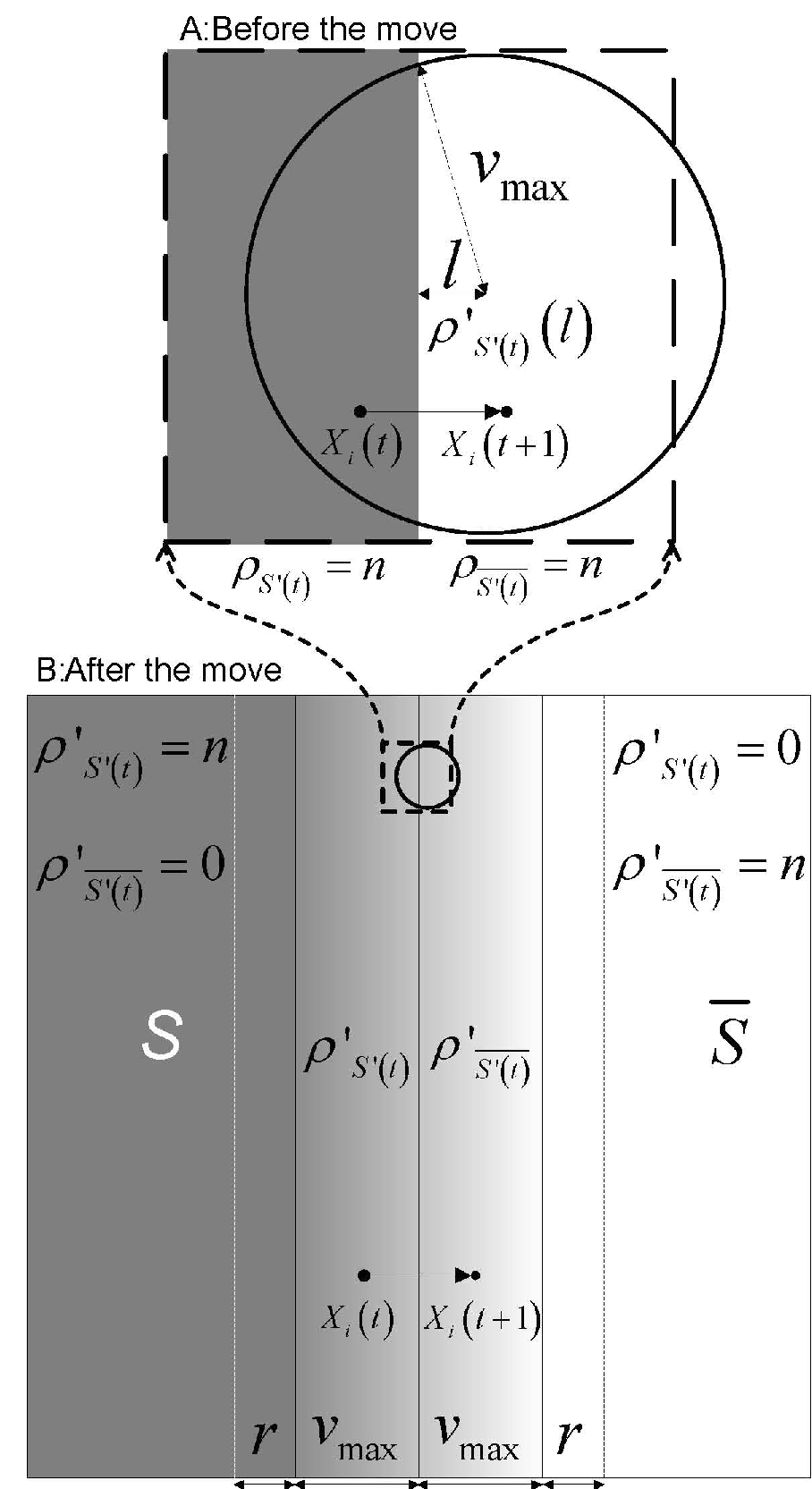}
\caption {Velocity Constrained Mobility} \label{velocity}
\end{figure}

\begin{figure}[h] \centering
\includegraphics[width=0.40\textwidth]{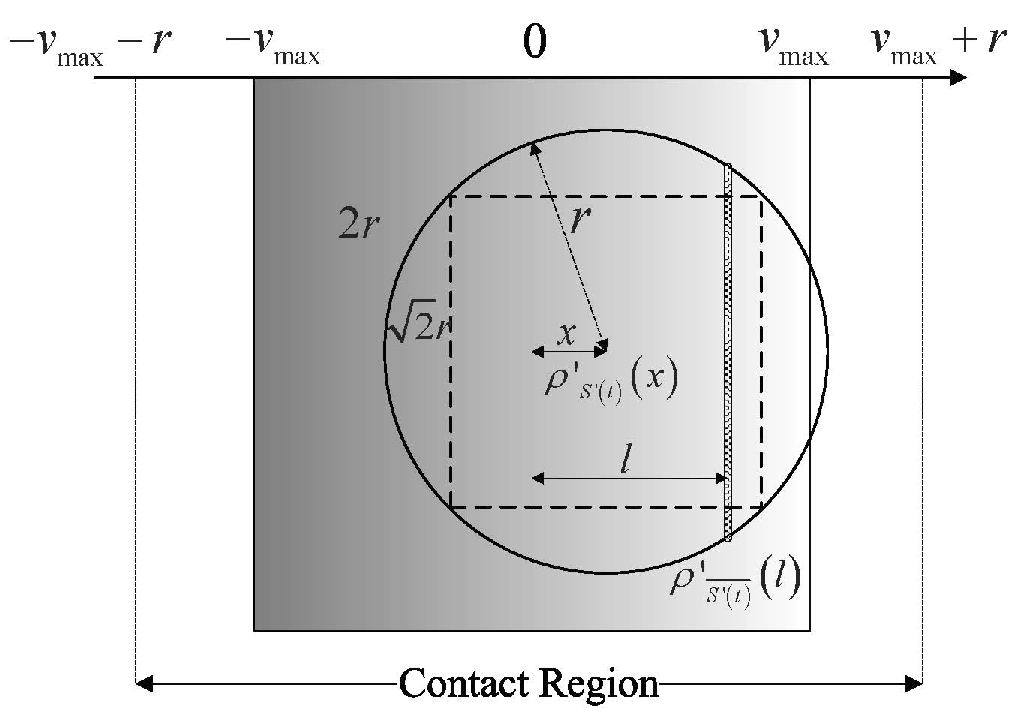}
\caption {Calculating the Number of Contact Pairs in Velocity Constrained Mobility} \label{velocity3}
\end{figure}

The contact region with the above bottleneck segmentation is the $2
\times \left( {v_{\max }  + r} \right)$ wide vertical strip in the center. All nodes outside this region will not contribute
to $N_{S'} \left( {t + 1} \right)$.

The number of contact pairs after the move can be calculated according to Fig.~\ref{velocity3}. The center of the circle with radius $r$ is $x$-distance away from the middle line. For node $i\in S'(t)$ located at the
center, the number of nodes that it can contact is equal to
the number of nodes belonging to $\overline{S'(t)}$ in the circle. Since the density of nodes belonging to $\overline{S'(t)}$ at positions $l$ away from the middle line
is ${\rho '_{\overline {S'} } \left( l \right)}$, the number of
nodes that $i$ can `push' information to is $\int\limits_{x - r}^{x
+ r} {\rho '_{\overline {S'} } \left( l \right) 2 \sqrt {r^2  -
\left( {l - x} \right)^2 } dl}$. Taking all nodes with message in
the contact region into consideration, the expected number of
contact pairs after the move is
\begin{align}\label{calculus}
& \mathbb{E}_Q \left[ {{N_{S'}}\left( {t + 1} \right)} \right]  \nonumber \\
  = & \int\limits_{ - v_{\max }  - r}^{v_{\max }  + r} {\rho '_{S'} \left( x \right)\int\limits_{x - r}^{x + r} {\rho '_{\overline {S'}} \left( l \right)2\sqrt {r^2  - \left( {l - x} \right)^2 } dldx} }.
 \end{align}

Since $S'\left( t \right)$ and $\overline {S'\left( t \right)}$ here
is the bottleneck segmentation that minimize the conductance, the
mobile conductance is ${\Phi _m}\left( Q \right)
= \frac{2}{{{n^2}\pi {r^2}}} \mathbb{E}_Q \left[ {{N_{S'}}\left( {t + 1} \right)} \right]$. According to the calculation in Appendix \ref{Evaluation}, we can obtain the
results in \emph{Theorem \ref{theorem-velocity}}.
\end{proof}

\emph{Remarks:} \emph{Theorem \ref{theorem-velocity}} indicates that, when ${v_{\max }=O(r)}$, $\Phi _m = \Theta \left(
{r } \right)$, and the spreading time scales as $O(\log n/r)$, which degrades to the static case;
when ${v_{\max }=\omega(r)}$, $\Phi _m = \Theta \left(
v_{\max } \right)$, and the spreading time scales as $O(\log n/ v_{\max })$, which improves over the static case and approaches the optimum when $v_{\max }$ approaches $\Theta(1)$. These observations are further verified through the simulation results below.



\subsubsection{One-dimensional Mobility}
%
\begin{theorem}\label{theorem-onedim} For the one-dimensional area constrained
mobility model, where among the $n$ nodes, $n_V$ nodes only move
vertically and $n_H$ nodes only move horizontally, The mobile
conductance $\Phi _m=\frac{{{n_V^2+n_H^2}}}{n^2}{\Phi _s} +  {\frac{{{n_V}{n_H}}}{{{n^2}}}} $.
\end{theorem}

\begin{proof} Denote the subset of V-nodes as
$S_V$, and the subset of H-nodes as $S_H$. Similar to the partially random mobility case, the calculation of the expected number of contact pairs is decomposed into four groups as follows.
\begin{align}\label{One-dim-Contact-pairs}
&\mathbb{E}_Q \left[ {{N_{S'}}\left( {t + 1} \right)} \right] \nonumber\\
= &\mathbb{E}_Q \left[ {\begin{array}{*{20}{l}}
{\sum\limits_{i \in {S_V} \cap S'\left( t \right),j \in {S_V} \cap \overline {S'\left( t \right)} } {{I_{ij}}\left( {t + 1} \right)} }\\
{ + \sum\limits_{i \in {S_H} \cap S'\left( t \right),j \in {S_H} \cap \overline {S'\left( t \right)} } {{I_{ij}}\left( {t + 1} \right)} }\\
{ + \sum\limits_{i \in {S_V} \cap S'\left( t \right),j \in {S_H} \cap \overline {S'\left( t \right)} } {{I_{ij}}\left( {t + 1} \right)} }\\
{ + \sum\limits_{i \in {S_H} \cap S'\left( t \right),j \in {S_V} \cap \overline {S'\left( t \right)} } {{I_{ij}}\left( {t + 1} \right)} }
\end{array}} \right].
\end{align}

Consider the first term, the expected number of contact pairs within V-nodes. Because all nodes in this case follow a one-dimensional ``fully
random" mobility model on their corresponding vertical paths, this number remains unchanged after the move. Therefore, the bottleneck
segmentation is the same as in the static
case, i.e. letting all V-nodes on the left half belong to
$S'\left( t \right)$ and those on the right half belong to
$\overline {S'\left( t \right)}$. However, the densities of the
V-nodes on both halves are $n_V$, instead of $n$. With respect to this bottleneck segmentation, the first term of \eqref{One-dim-Contact-pairs}, translated into \eqref{mobile-conductance}, gives $\left( {\frac{{{n_V}}}{n}} \right)^2 {\Phi _s}$.
Analogously, the bottleneck segmentation for the second term is formed by letting all H-nodes on the upper half belong to $S'\left( t \right)$ and those on the bottom half belong to $\overline {S'\left( t \right)}$, which contributes $\left( {\frac{{{n_H}}}{n}} \right)^2{\Phi _s}$ to the mobile conductance. These two bottleneck segmentations can be combined.

\begin{figure}[h] \centering
\includegraphics[width=0.35\textwidth]{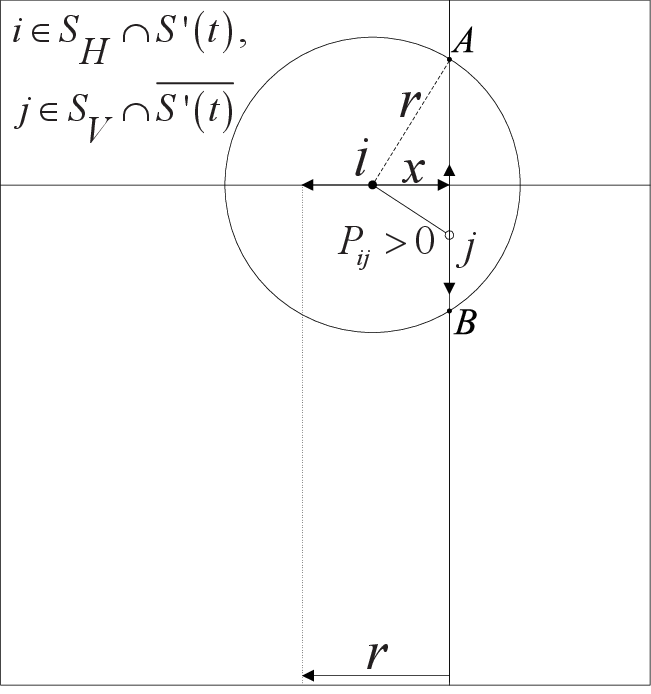}
\caption {Contact Probability of V-node and H-node in One-dimensional Area Constrained Mobility} \label{One-dim}
\end{figure}

Now we move on to the latter two terms of \eqref{One-dim-Contact-pairs}.
One key observation is that the contact probability between one V-node and one
H-node is independent of the positions of their paths in the unit square. To see this, let us check
Fig.~\ref{One-dim}: for any $i  \in S_H  \cap S'\left(
t \right)$ located $x$ away from the vertical path of $j \in S_V
\cap \overline {S'\left( t \right)}$, the probability that $(i,j)$
is a contact pair is the proportion of the chord length $|AB|$ over
the unit side length. Taking the integral over all $i$'s possible
positions on the horizontal path, their contact probability ${p_{_{H - V}}}$ is given by
\begin{equation*}
{p_{_{H - V}}} = \int_{ - r}^r {2\sqrt {{r^2} - {x^2}} } dx = \pi r^2.
\end{equation*}

Similarly, the contact probability between any $i \in S_V  \cap S'\left( t \right)$ and $j \in S_H  \cap \overline {S'\left( t \right)}$, ${p_{_{V - H}}}$ is also $\pi r^2$. Thus, the latter two terms can be evaluated as ${p_{_{V - H}}}\left| {{S_V} \cap S'\left( t \right)} \right|\left| {{S_H} \cap \overline {S'\left( t \right)} }\right|$ and ${p_{_{H - V}}}\left| {{S_H} \cap S'\left( t \right)} \right|\left| {{S_V} \cap \overline {S'\left( t \right)} } \right|$, respectively,
which are both independent of the segmentation before the move.

%

To sum up, the mobile conductance for the one-dimensional mobility model is:
\begin{align}\label{one-dim-conduct}
&{\Phi _m}\nonumber ={\left( {\frac{{{n_V}}}{n}} \right)^2}{\Phi _s} + {\left( {\frac{{{n_H}}}{n}} \right)^2}{\Phi _s}\nonumber \\
&{\footnotesize{+ \mathop {\min }\limits_{\scriptstyle S'\left( t \right) \subset V\atop \scriptstyle \left| {S'\left( t \right)} \right| < n/2} \left\{ {\frac{{P\left( n, r \right)}}{{\left| {S'\left( t \right)} \right|}}\left( \begin{array}{l}
{p_{_{H - V}}}\left| {S'\left( t \right)} \right|\frac{{{n_H}}}{n}\left| {\overline {S'\left( t \right)} } \right|\frac{{{n_V}}}{n}\\
 + {p_{_{V - H}}}\left| {S'\left( t \right)} \right|\frac{{{n_V}}}{n}\left| {\overline {S'\left( t \right)} } \right|\frac{{{n_H}}}{n}
\end{array} \right)} \right\}}}\nonumber \\
&= \frac{{{n_V^2+n_H^2}}}{n^2}{\Phi _s} + \mathop {\min }\limits_{\scriptstyle S'\left( t \right) \subset V \atop
\scriptstyle \left| {S'\left( t \right)} \right| < n/2} \left\{ {\frac{{2{n_V}{n_H}\left| {\overline {S'\left( t \right)} } \right|}}{{{n^3}}}} \right\}\nonumber \\
&= \frac{{{n_V^2+n_H^2}}}{n^2}{\Phi _s} +  {\frac{{{n_V}{n_H}}}{{{n^2}}}}.
\end{align}
\end{proof}

\emph{Remarks:} We can see that, when all nodes move in one
direction, the mobile conductance is the same as the static
case. On the contrary, when half (or a constant proportion) of the nodes are V-nodes and the other half are H-nodes, the mobile conductance achieves its maximum of $\Theta \left( 1 \right)$, the same order as in the fully random mobility model. The implication is that multidirectional movement spreads information faster than unidirectional movement.

\subsubsection{Two-dimensional Mobility}

\begin{theorem}\label{theorem-twodim} For the two-dimensional area-constrained mobility model with mobility capacity $r_c$, the mobile conductance scales as  $\Theta \left( \max \left( r_c,r \right) \right) $.
\end{theorem}

\begin{proof} Denote by $H_S \triangleq \left\{ i_h, i \in S'\left(t\right) \right\}$ the set of home points for $S'\left(t\right)$, and $H_{\overline S }\triangleq\left\{ i_h, i \in \overline {S'\left(t\right)} \right\}$ the set of home points for $\overline {S'\left(t\right)}$. Let $X_{i_h}$ and $X_{j_h}$ denote the positions of home points $i_h$ and $j_h$, then $i$ and $j$ can possibly move to positions within a distance of $r$ only if their home points are within a distance of $2r_c+r$, i.e., $\mathbb{E}_Q \left[ {{I_{ij}}\left( {t + 1} \right)} \right] > 0$ only if $\left| {X_{i_h}- X_{j_h}} \right| < 2r_c+r$.
This is similar to the velocity constrained mobility model, except that the node's position before the move $X_i(t)$ is replaced by the position of its home point $X_{i_h}$, and $v_{\max}$ replaced by $r_c$.

We now show that the two-dimensional area-constrained mobile conductance can be obtained similarly to the velocity constrained mobile conductance.

\begin{enumerate}
\item Here the node positions of $S'\left( t \right)$ and $\overline {S'\left( t \right)}$ after the move are not conditioned on their positions before the move, but determined by the positions of their home points. When calculating the expected number of contact pairs, $H_S$ and $H_{\overline S }$ play the same roles as $S'\left( t \right)$ and $\overline {S'\left( t \right)}$ before the move, respectively.

\item The mobility capacity $r_c$ has the same effect on information spreading as the maximal velocity $v_{\max}$ in the velocity constrained model, both of which set a limit on the nodes' moving ability.

\end{enumerate}

Instead of finding the bottleneck segmentation between $S'\left( t \right)$ and $\overline {S'\left( t \right)}$ before the move as in the velocity constrained model, we need to find the bottleneck segmentation between the home points: $H_S$ and $H_{\overline S }$. Since the home points also form a random geometric graph, it can be shown that the bottleneck segmentation is formed by dividing the home points into two halves using a straight vertical line bisecting the unit square, as illustrated in Fig.~\ref{Two-dim2}.

\begin{figure}[h] \centering
\includegraphics[width=0.35\textwidth]{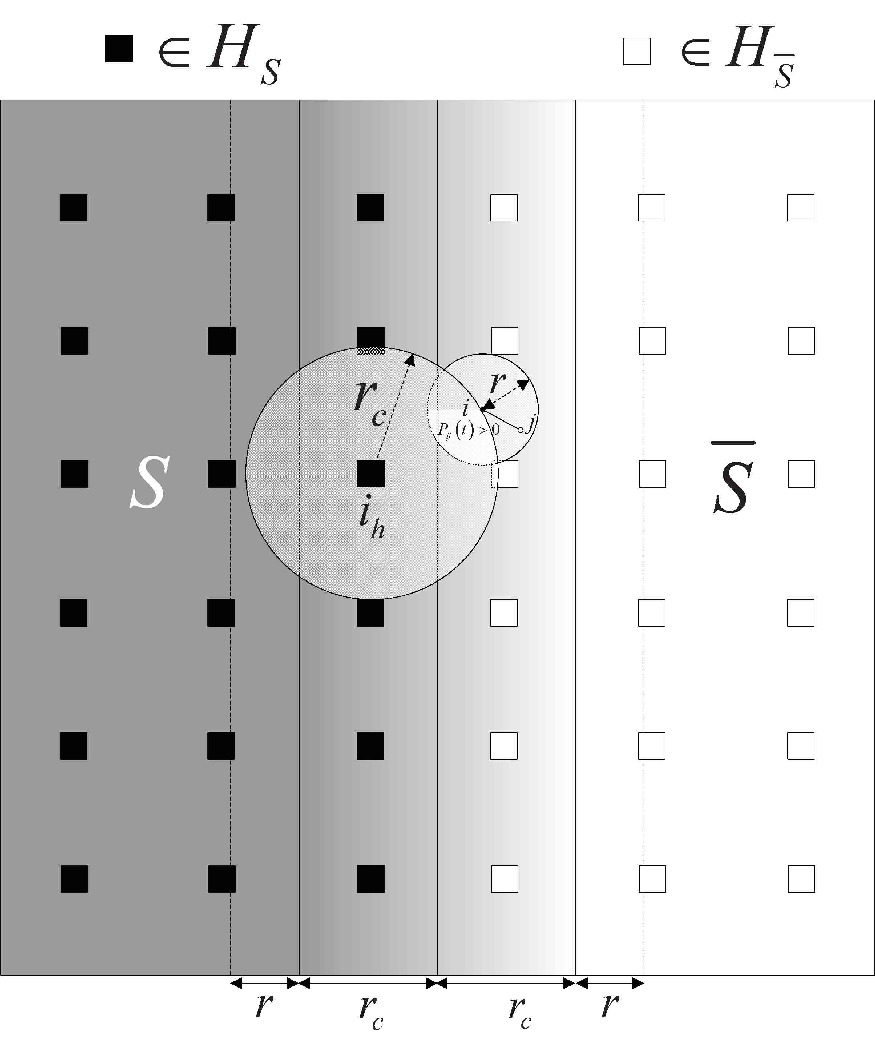}
\caption {Two-dimensional Area Constrained Mobility}
\label{Two-dim2}
\end{figure}

Therefore, we may follow the same line of evaluating the velocity constrained mobile conductance to obtain the two-dimensional area-constrained mobile conductance. The only difference is that $v_{max}$ in \eqref{Velocity-conduct} is replaced with the mobility capacity $r_c$, which leads to the final results in \emph{Theorem \ref{theorem-twodim}}.
\end{proof}

\emph{Remarks:} When the mobility capacity is much greater than the transmission radius, the mobile conductance is dominated by the mobility capacity, i.e.
$\Phi _m = \Theta \left( {r_c } \right)$. When the mobility capacity is much smaller than the transmission radius, the mobile conductance is dominated by the transmission radius, i.e. $\Phi _m = \Theta \left( {r } \right)$, as in static networks. The similarity between this model and the velocity constrained mobility model is worth noting.

\subsection{Simulation Results}
We have conducted large-scale simulations to verify the correctness and accuracy of the derived theoretical results. In our simulation, $n$ (up to 20,000) nodes are randomly deployed on a unit square and move according to certain mobility models, as described in Section \ref{problem-formulation}. The transmission radius $r\left( n \right)$ is set as $\sqrt {\frac{{{C_0}\log n}}{n}}$ with ${C_0} = \frac{8}{\pi }$ \cite{connectivity-constant2}. The spreading time is measured by the number of time slots. For each curve, we simulate one thousand Monte-Carlo rounds and present the average.

The spreading time results for static networks and fully random mobile networks are shown in Fig.~\ref{Tspr_PartialScaling} to Fig.~\ref{Tspr_TwoDimScaling} as the upper and lower bounds. We observe that the spreading time in mobile networks is significantly reduced, and as network size $n$ grows, the static spreading time increases much faster than the mobile counterpart. The bottommost curve (fully random mobility) grows in a trend of $\log n$ (note that the x-axis is on the log-scale), which confirms \emph{Theorem 2}.

Fig.~\ref{Tspr_PartialScaling} further confirms our remarks on \emph{Theorem \ref{theorem-partial}}. When the proportion of mobile nodes is a constant ($0.1$), the corresponding curve exhibits a slope almost identical to that for the fully random model. We also observe that $k =  \Theta \left( \sqrt{n\log n} \right)$ is a breaking point, below which ($k = \Theta \left( \sqrt{\log n} \right)$) the performance degrades to the static case.

Fig.~\ref{Tspr_VelocityScaling} confirms our remarks on \emph{Theorem \ref{theorem-velocity}}. When $v_{\max}=0.1$, the corresponding curve exhibits a slope almost identical to that for the fully random model. We also observe that $v_{\max}=\Theta \left( r \right)$ is a breaking point: velocity that is lower ($v_{\max}=o(r)=\Theta \left(\sqrt{\frac{1}{n}}\right)$) leads to a performance similar to the static case.

The spreading time results for the one-dimensional area constrained mobility model is shown in Fig.~\ref{Tspr_OneDimScaling}, which exhibit slopes almost identical to that for the fully random model. It is also shown that when half of the nodes are V-nodes and the other half are H-nodes, the best performance is achieved.

Fig.~\ref{Tspr_TwoDimScaling} confirms our remarks on \emph{Theorem \ref{theorem-twodim}}. When $r_c=0.1$, the corresponding curve exhibits a slope almost identical to that for the fully random model. We also observe that $r_c=\Theta \left( r \right)$ is a breaking point, a mobility capacity ($r_c=\Theta \left( \sqrt{\frac{1}{n}} \right)$) below which leads to a performance similar to the static case. Also note the similarity between this figure and Fig.~\ref{Tspr_VelocityScaling}.

\begin{figure}[h] \centering
\includegraphics[width=0.5\textwidth]{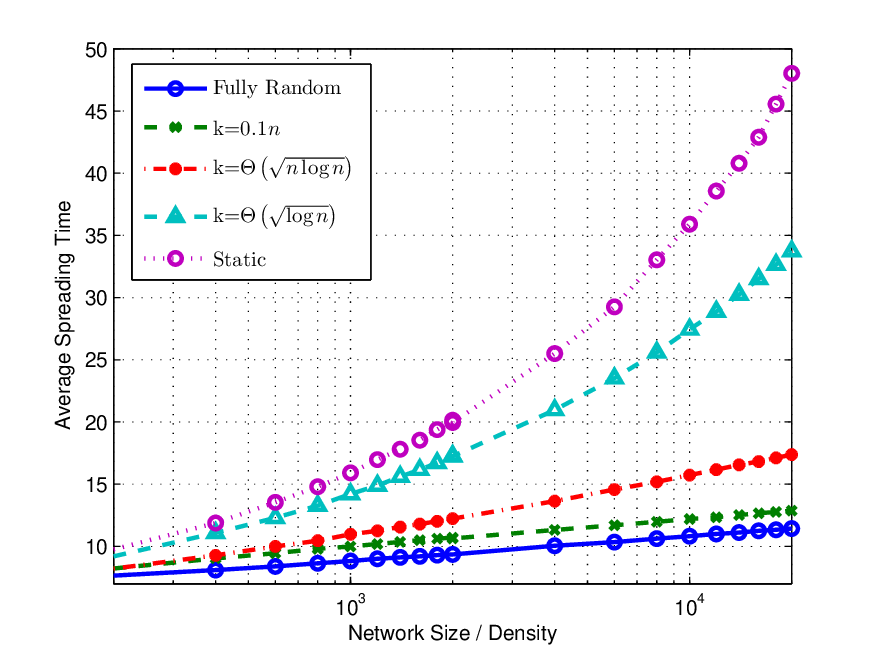}
\caption {Average Spreading Time under the Partially Random Mobility Model}
\label{Tspr_PartialScaling}
\end{figure}

\begin{figure}[h] \centering
\includegraphics[width=0.5\textwidth]{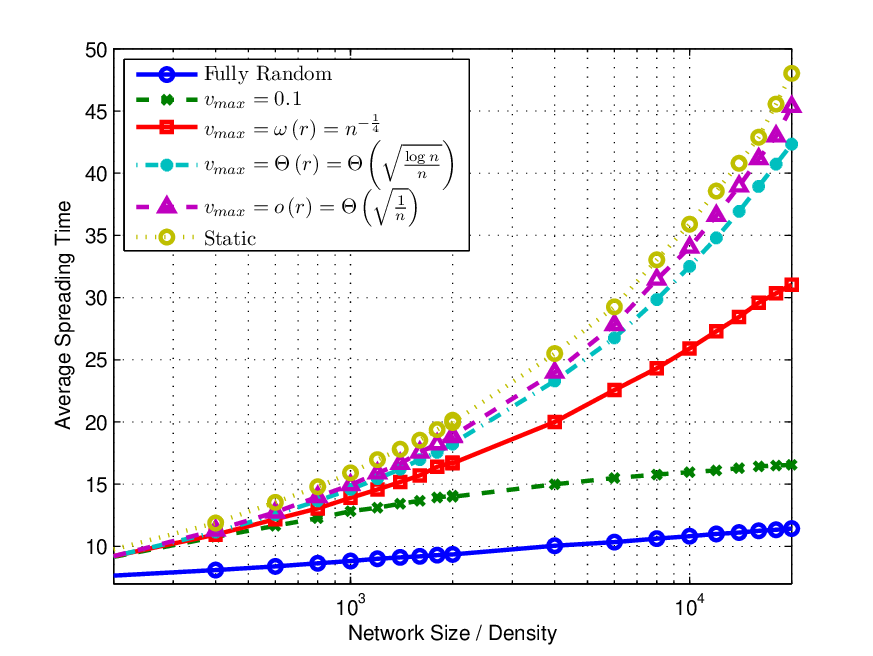}
\caption {Average Spreading Time under the Velocity-Constrained Mobility Model}
\label{Tspr_VelocityScaling}
\end{figure}

\begin{figure}[h] \centering
\includegraphics[width=0.5\textwidth]{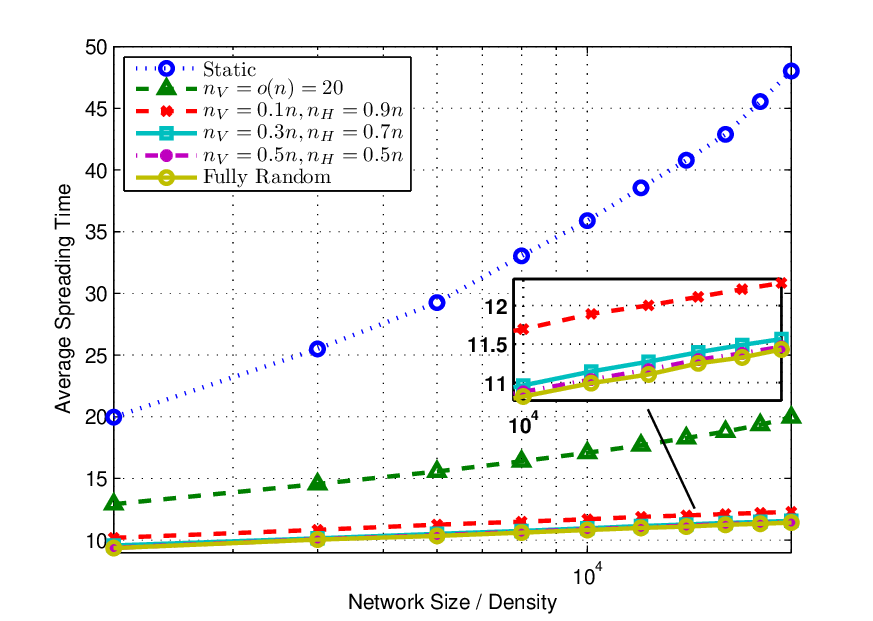}
\caption {Average Spreading Time under the One-Dimensional Area Constrained Mobility Model}
\label{Tspr_OneDimScaling}
\end{figure}

\begin{figure}[h] \centering
\includegraphics[width=0.5\textwidth]{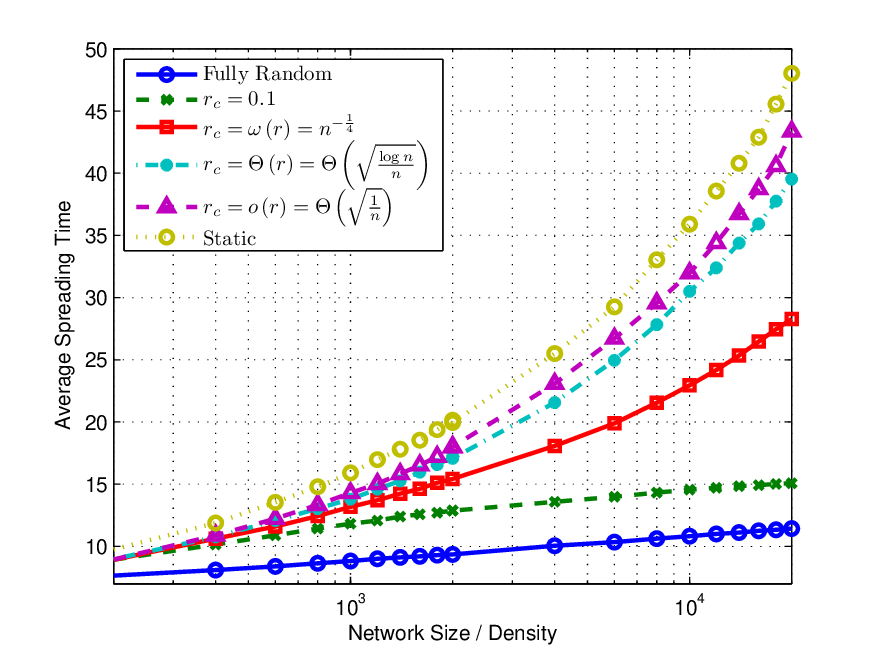}
\caption {Average Spreading Time under the Two-Dimensional Area Constrained Mobility Model}
\label{Tspr_TwoDimScaling}
\end{figure}

%





\section{Conclusions and Future Work}\label{conclusion}
\noindent In this paper, we analyze information spreading in
mobile networks, based on the proposed move-and-gossip information spreading model. For a dynamic graph that is connected under mobility, i.e., $v_{\max}+r=\Omega(\sqrt{\log n/n})$, we have derived a general expression for the information spreading time by gossip algorithms in terms of the newly defined metric mobile conductance, and shown that mobility can significantly speed up information spreading. This common framework facilitates the investigation and comparison of different mobility patterns and their effects on information dissemination.

In our current definition of mobile conductance, it is assumed that in each step, there exist some contact pairs between ${S'\left( t \right)}$ and $\overline {S'\left( t \right)}$ after the move. In extremely sparse networks (depending on the node density and transmission radius), we may have $E_{Q} \left[ {{N_{S'}}\left( {t + 1} \right)} \right]=0$. Let $T_{m}(i,j)\triangleq \inf \{t: j \in {\cal N}_i(t+1) \}$ be the first meeting time of nodes $i$ and $j$. We plan to extend the definition of mobile conductance to the scenario with $E[T_{m}(i,j)]<\infty$.




%


\appendices
\section{Bottleneck Segmentation for Velocity Constrained Mobility Model}\label{Lemma1}
\begin{lemma}\label{velocity-segmentation}
The bottleneck segmentation in mobile conductance evaluation under the RGG and velocity constrained
mobility model is a vertical straight line bisecting the unit square.
\end{lemma}
\begin{proof}
Given an arbitrary $S'\left(t\right)$ satisfying ${\left|S'\left(t\right)\right|}={n_0}<n/2$, it is necessary to minimize $\mathbb{E}_Q\left[ {{N_{S'}}\left( {t + 1} \right)} \right]$ in order to achieve the minimum of mobile conductance defined in
\eqref{mobile-conductance-approx}. To achieve the minimum, we first argue the existence of a border between $S'\left( t \right)$ and $\overline {S'\left( t \right)}$ and then determine the border type.

The expected number of contact pairs after the move can be represented as
\begin{align} \label{ArgueABorder}
&\mathbb{E}_Q\left[ {{N_{S'}}\left( {t + 1} \right)} \right] \nonumber \\
= &\mathbb{E}_Q\left[ \begin{array}{l}
\sum\limits_{i \in S'\left( t \right),j \in V\left( t \right)} {{I_{ij}}} \left( {t + 1} \right) \nonumber \\
 - \sum\limits_{i \in S'\left( t \right),j \in S'\left( t \right)} {{I_{ij}}} \left( {t + 1} \right)
\end{array} \right]  \nonumber \\
=&\left| {S'\left( t \right)} \right|n\pi {r^2} - \mathbb{E}_Q\left[
{\sum\limits_{i \in S'\left( t \right),j \in S'\left( t \right)}
{{I_{ij}}} \left( {t + 1} \right)} \right].
\end{align}

Therefore minimizing $E_Q\left[ {{N_{S'}}\left( {t + 1} \right)} \right]$ is equivalent to maximizing the second term in \eqref{ArgueABorder}. To be specific, in the second term $E_Q\left[ {{I_{ij}}\left( {t + 1} \right)} \right] > 0$ only if $i$ and $j$ can possibly move to positions within a distance of $r$, i.e., $ \left| {X_i\left( t \right)-X_j\left( t \right)} \right| < 2v_{\max }+r$, and the maximum is reached when the number of such node pairs in ${S'\left(t\right)}$ is maximized. It is claimed that this maximum is achieved only when the nodes in $S'\left( t \right)$ occupy an exclusive region, say $\mathbb{R}$, which does not contain any nodes from $\overline {S'\left( t \right)}$. Otherwise, if there exists a node ${\mathord{\buildrel{\lower3pt\hbox{$\scriptscriptstyle\frown$}}
\over j} } \in \overline {S'\left( t \right)}$ inside $\mathbb{R}$, simply replacing ${\mathord{\buildrel{\lower3pt\hbox{$\scriptscriptstyle\frown$}}
\over j} }$ with any node ${\mathord{\buildrel{\lower3pt\hbox{$\scriptscriptstyle\smile$}}
\over j} }\in S'(t)$ on the border at least will not decrease the second term in \eqref{ArgueABorder}, given that nodes in RGG are uniformly distributed. Therefore, without loss of generality, we can restrict our attention to the scenarios where there is a closed border separating ${S'\left(t\right)}$ and $\overline {S'\left( t \right)}$, and the nodes at least $2v_{\max }+r$ away from the border cannot have meaningful contact after the move.
\begin{figure}[h] \centering
\includegraphics[width=0.35\textwidth]{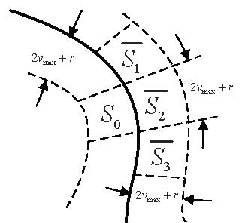}
\caption {A network cut in velocity constrained mobile networks}
\label{Velocity-Border}
\end{figure}

Denote by
${B\left( t \right)}$ the length of the border between ${S'\left(t\right)}$ and $\overline {S'\left( t \right)}$ and partition the area
near the border into bins of area (approximately) $S_b={\left(
{2{v_{\max }} + r} \right)^2}$ as shown in Fig.~\ref{Velocity-Border}. The
nodes in a bin can contact nodes in three bins on the opposite side.
For example, the nodes in ${S_0}$ can reach $\overline {{S_1}}$,
$\overline {{S_2}}$, and $\overline {{S_3}}$ after the move. Given
$v_{\max }$, $r$ and $n$, the number of possible contact pairs for
${S_0}$ is on the order of $\Theta \left( {{n^2}S_b^2} \right)$. The
total number of contact pairs after the move is
\begin{align*}
{N_{S'}}\left( {t + 1} \right) = \frac{{B\left( t
\right)}}{{2{v_{\max }} + r}}\Theta \left( {{n^2}S_b^2} \right).
\end{align*}

Therefore, the number of contact pairs after the move is
proportional to the length of the border before the move, i.e.
$N_{S'}(t+1) \propto B\left(t\right)$, and the mobile conductance
$\Phi_m \propto \frac {B\left(t\right)}
{\left|S'\left(t\right)\right|}$. Following the same argument as in the static case \cite{Conduct}, the ratio of $\frac
{B\left(t\right)} {\left|S\left(t\right)\right|}$ is minimized by a
vertical straight line bisecting the unit square. Therefore, the
mobile conductance under velocity constrained mobility is also
minimized through this \emph{bottleneck segmentation}.
\end{proof}

\section{Evaluation of Velocity Constrained Mobile Conductance}\label{Evaluation}
The accurate evaluation in \eqref{calculus} over the circle as shown in Fig.~\ref{velocity3} is
rather involved, therefore we loosen the requirement by only
calculating over the dashed square in the circle, as
illustrated in Fig.~\ref{velocity3}. Specifically, we replace \eqref{calculus} with
\begin{align*}
&\mathbb{E}_Q\left[ {{N_{S'}}\left( {t + 1} \right)} \right]\\
\cong & \int\limits_{ - {v_{\max }} - r}^{{v_{\max }} + r} {{{\rho '}_{S'}}\left( x \right)\int\limits_{x - \frac{r}{{\sqrt 2 }}}^{x + \frac{r}{{\sqrt 2 }}} {{{\rho '}_{\bar S'}}\left( l \right)\sqrt 2 rdldx} }.
\end{align*}

This will result in a smaller
mobile conductance, but the scaling law will not be affected in the
order sense. After some calculation, the mobile conductance is approximated by
\begin{align}\label{Velocity-conduct}
&\Phi _m \cong \left\{{\begin{array}{*{20}c}
  {\frac{1}{2}r + \frac{{v_{\max }^2 }}{{3r}},\quad \quad \quad \quad {\rm{for }}\quad v_{\max } \le \frac{1}{2} r }, \\
   { - \frac{{r^3 }}{{48v_{\max }^2 }} + \frac{{r^2 }}{{6v_{\max } }} + \frac{2}{3}v_{\max } ,\quad {\rm{for }}\quad v_{\max } > \frac{1}{2} r } . \\
\end{array}} \right.
\end{align}

\begin{biography}[{\includegraphics[width=1.0\textwidth] {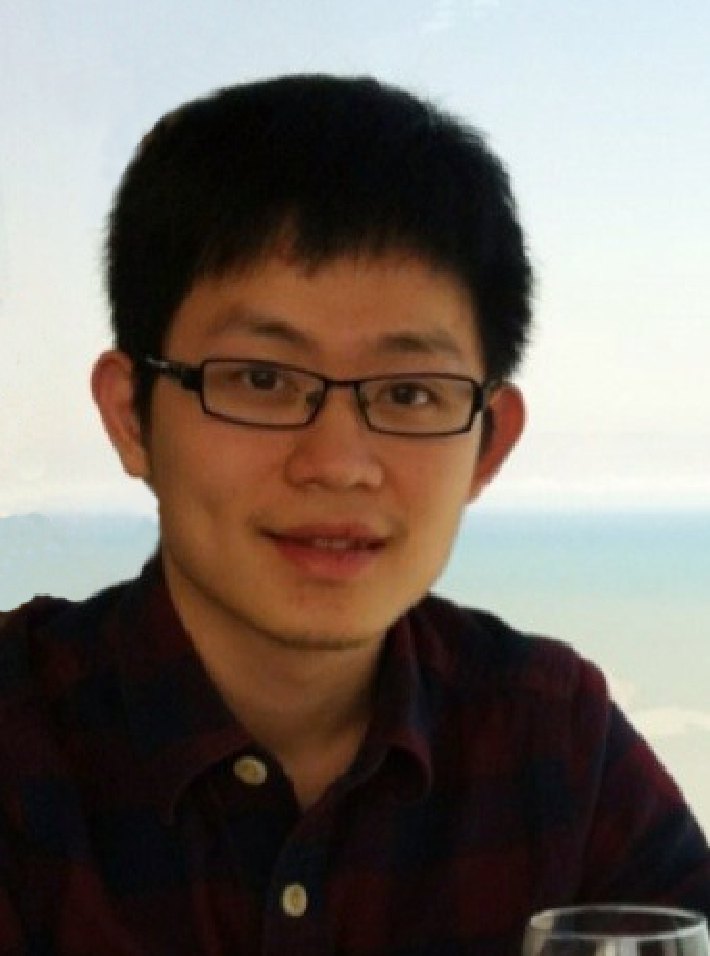}}]{Huazi Zhang}
[S'11] received the B.S. and Ph.D. degrees in electrical engineering from Zhejiang University, Hangzhou, China, in 2008 and 2013, respectively. From 2011 to 2013, he worked as a visiting scholar in the Department of Electrical and Computer Engineering, NC State University, Raleigh, NC.

His research interests include OFDM and MIMO systems, clustering and compressive sensing, coding techniques and cognitive radio networks. His recent interests are mobile networking, big data and social networks,

He has reviewed papers for IEEE journals such as Journal on Selected Areas in Communications,
Transactions on Signal Processing, Transactions on Wireless Communications, Transactions on Communications and Signal Processing Letters, and conferences such as INFOCOM, ICC, GlobeCom, etc. He is a student member of the IEEE.
\end{biography}

\begin{biography}[{\includegraphics[width=0.9\textwidth] {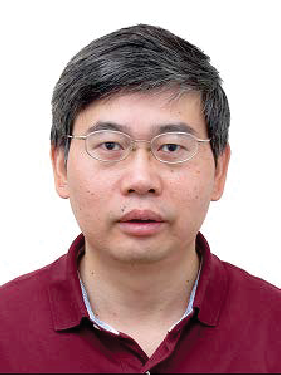}}]{Zhaoyang Zhang}
[M'02] received his B.Sc. and Ph.D. degrees from Zhejiang University
in 1994 and 1998, respectively. He is currently a full professor with the
Department of Information Science and Electronic Engineering, Zhejiang University. He was selected into the Supporting Program for New
Century Excellent Talents in University (NCET) by the Ministry of
Education, China, in 2009. His current research interests are
mainly focused on network information theory and
advanced coding theory, network signal processing,
cognitive radio networks and cooperative relay networks, etc., as well as their
applications in next generation wireless communication systems. He has authored
or co-authored more than 150 refereed journal and conference papers in the above areas. He has been actively serving as TPC co-chair, workshop co-chair, symposium co-chair, and TPC member for many international conferences. He is currently
serving as an Associate Editor for \emph{Wiley International
Journal of Communication Systems}.
\end{biography}

\begin{biography}[{\includegraphics[width=1.0\textwidth] {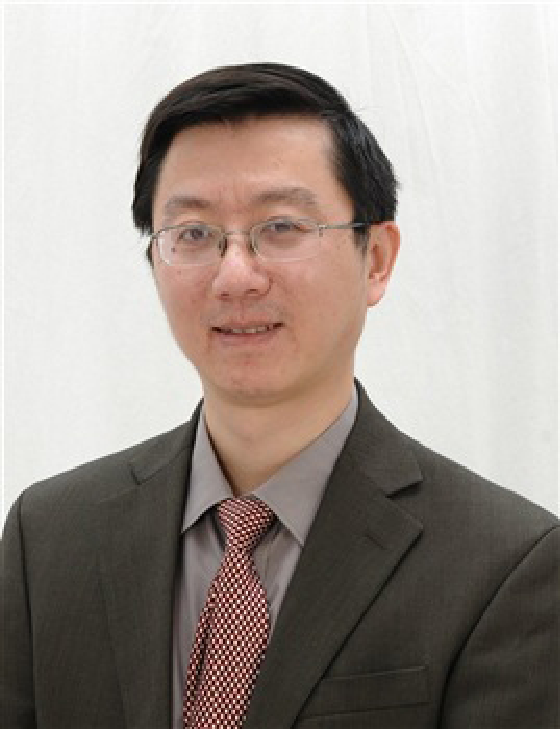}}]{Huaiyu Dai}
[M¡¯03, SM¡¯09] received the B.E. and M.S. degrees in electrical engineering from
Tsinghua University, Beijing, China, in 1996 and 1998, respectively, and the Ph.D. degree in electrical
engineering from Princeton University, Princeton, NJ in 2002.

He was with Bell Labs, Lucent Technologies, Holmdel, NJ, during summer 2000, and with AT\&T
Labs-Research, Middletown, NJ, during summer 2001. Currently he is an Associate Professor of
Electrical and Computer Engineering at NC State University, Raleigh. His research interests are in
the general areas of communication systems and networks, advanced signal processing for digital
communications, and communication theory and information theory. His current research focuses
on networked information processing and  crosslayer design in wireless networks, cognitive radio
networks, wireless security, and associated information-theoretic and computation-theoretic analysis.

He has served as an editor of IEEE Transactions on Communications, Signal Processing, and Wireless
Communications. He co-edited two special issues of EURASIP journals on distributed signal
processing techniques for wireless sensor networks, and on multiuser information theory and related
applications, respectively. He co-chairs the Signal Processing for Communications Symposium of
IEEE Globecom 2013, the Communications Theory Symposium of IEEE ICC 2014, and the Wireless
Communications Symposium of IEEE Globecom 2014.
\end{biography}
\end{document}